\documentclass[12pt]{amsart}
\usepackage[a4paper]{geometry}
\usepackage[utf8]{inputenc}
\usepackage[IL2]{fontenc}
\usepackage{amsfonts}
\usepackage{pdfsync}
\usepackage{amsmath}
\usepackage{amssymb}
\usepackage{amsthm}
\usepackage{mathrsfs}
\usepackage{latexsym}
\usepackage{enumerate}

\newtheorem{theorem}{Theorem}[section]

\newtheorem{lemma}[theorem]{Lemma}

\newtheorem{proposition}[theorem]{Proposition}
\newtheorem{remark}[theorem]{Remark}

\newcommand{\dd}{\mathrm{d}}
\newcommand{\R}{{\mathord{\mathbb R}}}

\newcommand{\N}{{\mathord{\mathbb N}}}

\newcommand{\supp}{{\mathop{\rm supp\ }}}

\newcommand{\dom}[1]{\mathrm{Dom}(#1)}

\newcommand{\vertiii}[1]{{\left\vert\kern-0.25ex\left\vert\kern-0.25ex\left\vert #1 
    \right\vert\kern-0.25ex\right\vert\kern-0.25ex\right\vert}}

\DeclareMathOperator*{\esssup}{ess\,sup}
\DeclareMathOperator*{\essinf}{ess\,inf}

\title{On an extension of the Iwatsuka model}
\author{Mat\v{e}j Tu\v{s}ek}
\date{\today}

\address{Department of Mathematics,
Faculty of Nuclear Sciences and Physical Engineering,
Czech Technical University in Prague,
Trojanova 13, 120\,00 Prague 2, Czech Republic}
\email{tusekmat@fjfi.cvut.cz}

\begin{document}

\begin{abstract}
  We prove absolute continuity for an extended class of two--dimen\-sional magnetic Hamiltonians that were initially studied by A. Iwatsuka. In particular, we add an electric field that is translation invariant in the same direction as the magnetic field is.  As an example, we study the effective Hamiltonian for a thin quantum layer in a homogeneous magnetic field.
\end{abstract}

\maketitle

\section{Introduction}
Consider a charged spin-less massive particle in a plane subject to  electric and magnetic fields that are both invariant with respect to the translations in $y$-direction. So if we denote these fields  $W$ and $B$, respectively, they are functions $W(x),\, B(x)$ of $x$ alone. Within the realm of non-relativistic quantum mechanics, the dynamics of the particle is governed by the following Hamiltonian
\begin{equation}\label{eq:ham_def}
 H=-\partial_{x}^{2}+(-i\partial_{y}+A_y(x))^2+W(x),
\end{equation}
where 
\begin{equation}\label{eq:vec_pot}
 A_y(x)=\int_0^{x}B(t)\dd t.
\end{equation}
Here we chose the Landau (asymmetric) gauge, put the reduced Planck constant and the ratio of the elementary charge to the speed of light equal to $1$, and fixed the particle’s mass to be $1/2$. The selfadjointness of this operator will be discussed at the beginning of Section \ref{sec:AC}.

It is long conjectured \cite{CFKS} that in the case without electric field ($W=0$), $H$ is purely absolutely continuous, i.e. $\sigma(H)=\sigma_{ac}(H)$, as soon as $B$ is non-constant. The conjecture was motivated by a seminal paper by A.~Iwatsuka \cite{Iw_85}. In acknowledgement of his achievement, the model described by  \eqref{eq:ham_def} bears his name.  He proved the absolute continuity of $H$ under the following additional pair of assumptions 
\begin{description}
 \item[AS1] $B\in C^\infty(\R;\R)$, and there exist constants $M_\pm$ such that $0<M_{-}\leq B\leq M_{+}$.
 \item[AS2] and either of the following holds
 \begin{description}
 \item[AS2a] $\limsup_{x\to -\infty}B(x)<\liminf_{x\to+\infty}B(x)$\\ or $\limsup_{x\to+ \infty}B(x)<\liminf_{x\to-\infty}B(x)$
 \item[AS2b]  $B$ is constant for all $|x|$ sufficiently large but non-constant on $\R$, and there exists $x_0$ such that $B'(x_{-})B'(x_{+})\leq 0$ for all $x_{-}\leq x_0\leq x_{+}$.
 \end{description}
\end{description}
In fact, \textbf{AS2b} may be relaxed to
\begin{description}
 \item[\phantom]
\begin{description}
 \item[AS2c] $B$ is non-constant and there exists a point $x_0$ such that for all $x_1,\, x_2$ with $x_1\leq x_0\leq x_2$ one has either $B(x_1)\leq B(x_0)\leq B(x_2)$ or $B(x_1)\geq B(x_0)\geq B(x_2)$,
\end{description}
\end{description}
as was proved by M.~M\v{a}ntoiu and R.~Purice \cite{Ma_Pu_97}. Remark that there is also some overlap of \textbf{AS2c}  with \textbf{AS2a}.

Another nice result concerning a variation of the magnetic field that is compactly supported (in $x$-variable) was given by P.~Exner and H.~Kova\v{r}\'{i}k \cite{ExKo_00}. They proved that $H$ is purely absolutely continuous if
\begin{description}
 \item[AS3] $B(x)=B_0+b(x)$, where $B_0>0$ and $b$ is bounded, piecewise continuous and compactly supported
 \item[AS4] and either of the following holds
 \begin{description}
 \item[AS4a] $b$ is nonzero and does not change sign
 \item[AS4b] let $[a_l,a_r]$ be the smallest closed interval that contains $\supp{b}$; there are $c,\delta>0$ and $m\in\N$ such that 
 $|b(x)|\geq c(x-a_l)^m \text{ or } |b(x)|\geq c(a_r-x)^m$
 for all $x\in[a_l,a_l+\delta)$ or $x\in(a_r-\delta,a_r]$, respectively.
 \end{description}
\end{description}

In this paper we generalize \textbf{AS2a} to the case when the electric field is switched on relaxing \textbf{AS1} simultaneously. First, we will need some notation to lighten the text. For any $f\in L^\infty(\R;\R)$, let us define
\begin{align*}
 \underline{f}_{+}&=\sup_{a\in\R}\essinf_{t\in(a,+\infty)}f(t) & \overline{f}_{+}&=\inf_{a\in\R}\esssup_{t\in(a,+\infty)}f(t)\\
 \underline{f}_{-}&=\sup_{a\in\R}\essinf_{t\in(-\infty,a)}f(t) & \overline{f}_{-}&=\inf_{a\in\R}\esssup_{t\in(-\infty,a)}f(t).
\end{align*}

\begin{theorem}\label{theo:main}
 Let $B,\, W\in L^\infty(\R;\R)$ be such that  either
 \begin{equation}\label{eq:ac_cond} 
  \underline{B}_{\pm}>0\,\wedge\, \underline{B}_{+}\geq\overline{B}_{-}\,\wedge\,(\overline{W}_{-}-\underline{W}_{+}<\underline{B}_{+}-\overline{B}_{-})
 \end{equation}
 or
 \begin{equation}\label{eq:ac_cond_alt}
 \underline{B}_{+}>0\,\wedge\, \overline{B}_{-}<0.
 \end{equation}
 Then $H$ is purely absolutely continuous. The same holds true if we interchange the $\pm$ indices everywhere in \eqref{eq:ac_cond} and \eqref{eq:ac_cond_alt}.
\end{theorem}

Let us stress that we do not require $B$ to be everywhere greater then some positive constant. In fact, under \eqref{eq:ac_cond_alt}, $B$ has to change its sign, and under \eqref{eq:ac_cond}, it may be negative on a compact set. As far as I can see, it is not possible to extend the proof of \cite{Ma_Pu_97} to include this case, nor to the case of non-zero $W$. On the other hand, the Iwatsuka's proof may be non-trivially modified to work under only slightly  stricter assumptions than those of Theorem \ref{theo:main}. In particular, one needs the derivative of $W$ to be in $L^\infty$. 

The Iwatsuka's strategy may be described as follows. First, decompose $H$ into the direct integral of one-dimensional operators with purely discrete spectrum. Then show that these fiber operators form an analytic family with simple and non-constant eigenvalues with respect to the quasi-momentum parameter. Although our proof, as well as all the proofs of the above mentioned results, follows this strategy, it differs in the method used when proving the last step, i.e. the non-constancy of the eigenvalues. Let us stress that this very step is typically the most difficult to prove. Whereas Iwatsuka needed some estimates on the growth of the eigenfunctions to show that the asymptotic behaviour of the eigenvalues in $\pm\infty$ is determined by that of the magnetic field, we derive the asymptotic behaviour of the eigenvalues directly using some comparison argument  based on the minimax principle combined with a norm-resolvent convergence result.

Iwatsuka's model with a non-zero electric field ($W\neq 0$) of a particular type has been studied before in \cite{ExJoKo_01}. There it was proved that $H_{\mathrm{L}}+\omega^2x^2$ remains purely absolutely continuous under a perturbation that is either a bounded function of $x$--variable only or a bounded periodic function of $y$--variable only. Here, $H_\mathrm{L}$ is the Landau Hamiltonian (\eqref{eq:ham_def} with $W=0$ and constant $B\neq 0$) and $\omega>0$.
For $H=H_{\mathrm{L}}+W$, where $W=W(x)$ is non-decreasing non-constant bounded function, it was proved that $\sigma(H)$ has a band structure and is purely absolutely continuous \cite{BrMiRa_11}. The asymptotic distribution of the discrete spectrum of $H$ under a bounded perturbation of constant sign that decays at infinity was investigated in the same paper. The same problem with $W=W(x)$ now being  periodic was addressed in \cite{MiRa_12}. In this case, the absolute continuity of $H$ was demonstrated only below a fixed but arbitrarily large energy when the magnetic field is strong enough. The asymptotic distribution of eigenvalues in spectral gaps was also studied for the case of the Iwatsuka Hamiltonian, essentially satisfying \textbf{AS1} and \textbf{AS2a}, with an additional electric potential that is either power-like decaying at infinity or compactly supported \cite{Mi_15}.

Transport properties of the Iwatsuka model are also of continuous interest. A non-constant translation invariant magnetic field acts as a magnetic barrier and gives rise to the so-called edge currents that are quantized \cite{DoGeRa_11}. When the barrier is sharp, the current-carrying states are well localized and stable with respect to various magnetic and electric perturbations \cite{HiSo_15,DoHiSo_14}. The latter paper deals solely with one the configurations when the magnetic field is constant on each of two complementary half-planes (the so-called magnetic steps). The analysis of the magnetic steps is completed in \cite{HiPoRaSu_16}. They had been studied before from a physicist's point of view in \cite{RePe_00}.

Finally, to demonstrate the richness of the topic, let us mention that a three-dimensional version of the Iwatsuka model was studied in \cite{Ya_08} and  random analogues of the Iwatsuka model were examined in \cite{LeWaWe_06}.

The paper is organized as follows. In Section \ref{sec:AC}, we introduce  operator \eqref{eq:ham_def} properly and prove that its spectrum is purely absolutely continuous under our assumptions. The proof relies on an abstract operator convergence result that is proved separately in Section \ref{sec:abstract_result}. Section \ref{sec:example} is devoted to a significant example that comes from the realm of the so--called  quantum waveguides. A particle confined in a very thin curved layer in an ambient constant magnetic field is effectively subjected to a non--constant magnetic field that is given by the projection into the normal direction to the layer \cite{KrRaTu_15}. Moreover, the non-trivial curvature gives rise to an additional attractive scalar potential. An interplay between magnetic and electric fields (that may be, in the discussed example, expressed solely in terms of geometric 
quantities) is reflected in our sufficient condition \eqref{eq:ac_cond}.

\section{Absolute continuity of the Hamiltonian}\label{sec:AC}
\subsection{Direct integral decomposition} \label{sec:decopm}

 By $H$ we mean the closure of $\dot{H}$ given by 
 \begin{align*}
  & \dot{H}=-\partial_{x}^{2}+(-i\partial_{y}+A_y(x))^2+W(x)\\
  &\dom{\dot{H}}=C_{0}^{\infty}(\R^2)\subset L^2(\R^2),
 \end{align*}
 where $W\in L^\infty(\R;\R)$ and $A_y$ is given by \eqref{eq:vec_pot} with $B\in L^\infty(\R;\R)$.
 $H$ is selfadjoint \cite{LeSi_81} (we refer to this paper whenever essential selfadjointness is mentioned) and commutes with the translations in $y$-direction. In \cite{Iw_85}, it was demonstrated that $H$ is unitarily equivalent to a direct integral in $L^2(\R_\xi;L^2(\R_x))$ of selfadjoint operators  $\{H[\xi],\ \xi\in\R\}$ with the following action
$$H[\xi]=-d^{2}_{x}+(\xi+A_y(x))^2+W(x).$$
For any $\xi\in\R$, $C_{0}^{\infty}(\R)$ is a core of $H[\xi]$.

To prove the absolute continuity of $H$, it is sufficient to show that \cite[Theo. XIII.86]{RS4} 
\begin{enumerate}[1)]
 \item The family $\{H[\xi]|\, \xi\in\R\}$ is analytic in $\xi$.
 \item For all $\xi\in\R$, $H[\xi]$ has compact resolvent.
 \item If we number the eigenvalues of $H[\xi]$ in strictly increasing order as $\lambda_n[\xi],\ n\in\N$, then every $\lambda_n[\xi]$ is simple and no $\lambda_n[\xi]$ is constant in $\xi$.
\end{enumerate}
For any $\xi_0\in\R$, we may write
$$H[\xi]=H[\xi_0]\dotplus p_{\xi},$$
where the quadratic form $p_{\xi}$ reads
\begin{equation}\label{eq:xi_pert}
p_{\xi}(\psi)=(\xi-\xi_0)^2 \|\psi\|^2+2(\xi-\xi_0)\langle\psi,(\xi_0+A_{y})\psi\rangle.
\end{equation}
For all $\delta>0$, one easily gets
\begin{equation*}
\begin{split}
 |p_{\xi}(\psi)|&\leq (\xi-\xi_0)^2 (1 +\delta^{-1})\|\psi\|^2+\delta\|(\xi_0+A_{y})\psi\|^2\\
 &\leq (\xi-\xi_0)^2 (1+\delta^{-1})\|\psi\|^2 +\delta\langle\psi,H[\xi_0]\psi\rangle+\delta\langle\psi,W_{-}\psi\rangle\\
 &\leq \left((\xi-\xi_0)^2 (1+\delta^{-1})+\delta\|W_{-}\|_\infty\right)\|\psi\|^2 +\delta\langle\psi,H[\xi_0]\psi\rangle,
\end{split}
\end{equation*}
where $W_{-}$ stands for the negative part of $W$.

Hence, $p_{\xi}$ is infinitesimally form bounded by $H[\xi_0]$.  This together with (\ref{eq:xi_pert}) implies that, $H[\xi]$ forms an analytic family of type (B). In particular $H[\xi]$ is an analytic family in the sense of Kato \cite{RS4}, which proves the first point.

Assuming either \eqref{eq:ac_cond_alt} or the first part of \eqref{eq:ac_cond}, we deduce that $|\lim_{x\to\pm\infty}A_y(x)|=+\infty$. This implies compactness of the resolvent of $H[\xi]$ \cite[Theo. XIII.67]{RS4}, i.e., the second condition. Consequently, the spectrum of $H[\xi]$ is purely discrete. Moreover, mimicking the proofs of \cite[Prop. 3.1, Lem. 2.3(i)]{Iw_85} (using results of \cite[\textsection 16]{Nai2} and a simple observation that any continuous regular distribution with almost everywhere non-negative weak derivative is everywhere non--decreasing), one may infer that all the eigenvalues of $H[\xi]$ are simple. Therefore, the first part of the third condition holds true, too.

\begin{remark}
 With these results in hand we may conclude that the singular continuous component in the spectrum of $H$ is empty \cite{FiSo_06}.
\end{remark}

The second part of the third condition is easy to verify under the assumption \eqref{eq:ac_cond_alt}. If $\underline{B}_{+}>0\wedge \overline{B}_{-}<0$ then $\lim_{x\to\pm\infty}A_y(x)=+\infty$. Using the minimax principle we obtain $\lim_{\xi\to+\infty}\lambda_n[\xi]=+\infty$. If $\underline{B}_{-}>0\wedge \overline{B}_{+}<0$ then $\lim_{x\to\pm\infty}A_y(x)=-\infty$ and $\lim_{\xi\to-\infty}\lambda_n[\xi]=+\infty$.

The rest of this section is devoted to the verification of the third condition under assumption \eqref{eq:ac_cond}. In particular, we always assume that $\underline{B}_{\pm}>0$.

\subsection{Some auxiliary results}

\subsubsection{Estimate on the potential}

\begin{lemma}\label{lem:pot_est}
 Let \eqref{eq:ac_cond} holds and $\varepsilon\in(0,\min\{\underline{B}_{+},\underline{B}_{-}\}/2)$. Then $\xi_\varepsilon\in\R$ exists such that, for all $\xi<\xi_\varepsilon$,
 \begin{equation}\label{eq:pot_est}
  \underline{V}_{\varepsilon,+}(x)+\underline{W}_{+}-\varepsilon\leq (\xi+A_y(x))^2 +W(x)\leq\overline{V}_{\varepsilon,+}(x)+\overline{W}_{+}+\varepsilon\quad (\text{a.e. }x).
 \end{equation}
 Similarly, $\tilde{\xi}_\varepsilon\in\R$ exists such that, for all $\xi>\tilde{\xi}_\varepsilon$,
 \begin{equation*}
  \underline{V}_{\varepsilon,-}(x)+\underline{W}_{-}-\varepsilon\leq (\xi+A_y(x))^2 +W(x)\leq\overline{V}_{\varepsilon,-}(x)+\overline{W}_{-}+\varepsilon\quad (\text{a.e. }x).
 \end{equation*}
 Here,
 \begin{align*}
 &\underline{V}_{\varepsilon,+}(x):=\begin{cases}
                                  (\underline{B}_{+}-2\varepsilon)^2 (x-x_\xi)^2&\text{ for  }x\geq -K_{\varepsilon}\\
                                  \big((B_{min}-\varepsilon)(x+K_\varepsilon)+(\underline{B}_{+}-2\varepsilon) (-K_\varepsilon-x_\xi)\big)^2&\text{ for }x<- K_{\varepsilon}
                                 \end{cases}\\
 &\overline{V}_{\varepsilon,+}(x):=\begin{cases}
                                  (\overline{B}_{+}+2\varepsilon)^2 (x-x_\xi)^2&\text{ for  }x\geq -K_{\varepsilon}\\
                                  \big((B_{max}+\varepsilon)(x+K_\varepsilon)+(\overline{B}_{+}+2\varepsilon) (-K_\varepsilon-x_\xi)\big)^2&\text{ for }x<- K_{\varepsilon}
                                 \end{cases} \\ 
 &\underline{V}_{\varepsilon,-}(x):=\begin{cases}
                                  (\underline{B}_{-}-2\varepsilon)^2 (x-x_\xi)^2&\text{ for  }x\leq K_{\varepsilon}\\
                                  \big((B_{min}-\varepsilon)(x-K_\varepsilon)+(\underline{B}_{-}-2\varepsilon) (K_\varepsilon-x_\xi)\big)^2&\text{ for }x>K_{\varepsilon}
                                 \end{cases}\\    
 &\overline{V}_{\varepsilon,-}(x):=\begin{cases}
                                  (\overline{B}_{-}+2\varepsilon)^2 (x-x_\xi)^2&\text{ for  }x\leq K_{\varepsilon}\\
                                  \big((B_{max}+\varepsilon)(x-K_\varepsilon)+(\overline{B}_{-}+2\varepsilon) (K_\varepsilon-x_\xi)\big)^2&\text{ for }x>K_{\varepsilon},
                                 \end{cases}                                 
 \end{align*}
where $B_{min}:=\min\{\underline{B}_{+},\underline{B}_{-}\}$, $B_{max}:=\max\{\overline{B}_{+},\overline{B}_{-}\}$, $K_\varepsilon>0$ is introduced below, and $x_\xi$ is the unique solution of $(\xi+A_y(x))=0$. 
(Uniqueness, for all $\xi$ with $|\xi|$ sufficiently large, is proved below, too.)
\end{lemma}
\begin{proof}
 We will only prove the first inequality in \eqref{eq:pot_est}. The remaining inequalities may be deduced in a similar manner.
 
 Since $\underline{B}_{\pm}>0$, $B>\varepsilon$ almost everywhere outside a compact subset of $\R$. Moreover, $A_y$ is absolutely continuous and $A'_{y}=B$ (a.e. $x$). In particular, we have $\lim_{x\to\pm\infty}A_y(x)=\pm\infty$ and $(\xi+A_y(x))=0$ has unique solution for all $\xi$ with $|\xi|$ sufficiently large. Let us denote this solution by $x_\xi$. Clearly, $\lim_{\xi\to -\infty}x_\xi=+\infty$.
 
 For a given $\varepsilon$, there exists $K_\varepsilon>0$ such that almost everywhere on $(K_\varepsilon,+\infty)$,
 \begin{equation*}
  \underline{B}_{+}-\varepsilon<B<\overline{B}_{+}+\varepsilon,\quad \underline{W}_{+}-\varepsilon<W<\overline{W}_{+}+\varepsilon,
 \end{equation*}
 and almost everywhere on $(-\infty,-K_\varepsilon)$,
 \begin{equation*}
  \underline{B}_{-}-\varepsilon<B<\overline{B}_{-}+\varepsilon,\quad \underline{W}_{-}-\varepsilon<W<\overline{W}_{-}+\varepsilon.
 \end{equation*}
 Let us stress that the choice of $K_\varepsilon$ depends solely on $\varepsilon$. Now we restrict ourselves to $\xi$ sufficiently negative so that $x_\xi>K_\varepsilon$, and we estimate
 \begin{align}
  & \xi+A_y(x)>(\underline{B}_{+}-2\varepsilon)(x-x_\xi)\geq 0 &\text{ for }x\geq x_\xi \label{eq:est1}\\
  & \xi+A_y(x)<(\underline{B}_{+}-2\varepsilon)(x-x_\xi)<0 &\text{ for }x\in(K_\varepsilon,x_\xi) \label{eq:est2}\\
  & W(x)>\underline{W}_{+}-\varepsilon &\text{ for }x\geq K_\varepsilon.\label{eq:est3}
 \end{align}
 
 We also have
 $$\sup_{x\in(-K_\varepsilon,K_\varepsilon)}|A_y(x)-A_y(K_\varepsilon)|\leq2K_\varepsilon\|B\|_\infty,\quad \esssup_{x\in(-\infty,K_\varepsilon)}|W(x)-(\underline{W}_{+}-\varepsilon)|<+\infty.$$
 Using these estimates together with the following observation,
 \begin{equation*}
  \lim_{\xi\to -\infty}\big( (\xi+A_y(K))-(\underline{B}_{+}-2\varepsilon)(K_\varepsilon-x_\xi)\big)=-\infty,
 \end{equation*}
 we infer that, for all sufficiently negative $\xi$, not only
 $\xi+A_y(x)<(\underline{B}_{+}-2\varepsilon)(x-x_\xi)<0$ 
 but
 \begin{equation}\label{eq:est4}
  (\underline{B}_{+}-2\varepsilon)^2(x-x_\xi)^2+\underline{W}_{+}-\varepsilon<(\xi+A_y(x))^2+W(x)
 \end{equation}
 on $(-K_\varepsilon,K_\varepsilon)$. Finally, by a similar reasoning, there exists $\xi_\varepsilon$ such that for all $\xi<\xi_\varepsilon$, $x_\xi>K_\varepsilon$ and \eqref{eq:est4} holds together with
 \begin{equation}\label{eq:est5}
 \big((\min\{\underline{B}_{+},\underline{B}_{-}\}-\varepsilon)(x+K_\varepsilon)+(\underline{B}_{+}-2\varepsilon) (-K_\varepsilon-x_\xi)\big)^2+\underline{W}_{+}-\varepsilon<(\xi+A_y(x))^2+W(x)
 \end{equation}
on $(-\infty, -K_\varepsilon)$.
 
 Putting \eqref{eq:est1}, \eqref{eq:est2}, \eqref{eq:est3}, \eqref{eq:est4}, and \eqref{eq:est5} together we arrive at
 $$\underline{V}_{\varepsilon,+}(x)+\underline{W}_{+}-\varepsilon\leq (\xi+A_y(x))^2 +W(x)\quad\text{ for }x\in\R.$$
\end{proof}

\subsubsection{Abstract convergence result}

\begin{theorem}\label{theo:abstract}
 Let $\{A[\alpha],\, \alpha\in(-\infty,+\infty]\}$ be a one parametric family of lower--bounded selfadjoint operators on $L^2(\Omega)$, where $\Omega\subset\R^n$ is open, with the following properties
 \begin{enumerate}[\upshape (i)]
  \item \label{prop:core} $C_{0}^{\infty}(\Omega)$ is a core of $A[\alpha]$ for all $\alpha\in(-\infty,+\infty]$.
  \item \label{prop:low_bound} There exist $C>0$ and $K,\,\alpha_0\in\R$ such that, for all $\alpha\geq\alpha_0$, $C A[+\infty]+K\leq A[\alpha]$.
  \item \label{prop:conv} For any compact set $\mathcal{K}\subset\Omega$, there exists $\alpha_{\mathcal{K}}$ such that, for all $\alpha\geq\alpha_{\mathcal{K}}$, $A[\alpha]|_{C_{0}^{\infty}(\mathcal{K})}=A[+\infty]|_{C_{0}^{\infty}(\mathcal{K})}$.
  \item \label{prop:comp} $A[+\infty]$ has compact resolvent.
 \end{enumerate}
 Then, for any $z\in\mathrm{Res}(A[+\infty])$ and $\varepsilon>0$, there exists $\alpha_{z,\varepsilon}$ such that for all $\alpha>\alpha_{z,\varepsilon}$, $z\in\mathrm{Res}(A[\alpha])$ and
 $$\|(A[\alpha]-z)^{-1}-(A[+\infty]-z)^{-1}\|<\varepsilon.$$
\end{theorem}

The proof is given separately in Section \ref{sec:abstract_result}.

\subsubsection{Comparison operators}\label{sec:comparison}
 Let $\omega,\, \tilde{\omega}>0$ and $x_0,\alpha\in\R$. The following differential operators on $L^2(\R)$
 \begin{align*}
  &\dot{H}_{\omega,\tilde{\omega}}[\alpha]:=\begin{cases}
                                   -d_{x}^{2}+\omega^2 (x-\alpha)^2 &\text{ for } x\geq x_0\\
                                   -d_{x}^{2}+\left(\tilde{\omega}(x-x_0)+\omega(x_0-\alpha)\right)^2 &\text{ for } x<x_0
                                   \end{cases}\\
  &\dot{H}_\omega[\alpha]= -d_{x}^{2}+\omega^2 (x-\alpha)^2
 \end{align*}
  defined on $C_{0}^{\infty}(\R)$ are essentially selfadjoint. We will denote their closures by $H_{\omega,\tilde{\omega}}[\alpha]$ and $H_{\omega}[\alpha]$, respectively. Let us introduce a unitary transform $U_\alpha:\ \psi(x)\mapsto\psi(x-\alpha)$. Then
  $$H_{\omega,\tilde{\omega}}[\alpha]=U_\alpha \tilde{H}_{\omega,\tilde{\omega}}[\alpha]U_{\alpha}^{*},\quad H_{\omega}[\alpha]=U_\alpha \tilde{H}_{\omega}U_{\alpha}^{*}$$
  with 
 \begin{align*}
  &\tilde{H}_{\omega,\tilde{\omega}}[\alpha]:=\begin{cases}
                                   -d_{x}^{2}+\omega^2 x^2 &\text{ for } x\geq x_0-\alpha\\
                                   -d_{x}^{2}+\left(\tilde{\omega}(x+\alpha-x_0)+\omega(x_0-\alpha)\right)^2 &\text{ for } x<x_0-\alpha
                                   \end{cases}\\
  &\tilde{H}_\omega= -d_{x}^{2}+\omega^2 x^2.
 \end{align*} 
 Remark that $\tilde{H}_\omega$ is just the harmonic oscillator Hamiltonian whose spectrum is very well known to be formed only by simple eigenvalues $(2n-1)\omega,\, n\in\N$. Due to unitary equivalence, $\sigma(H_{\omega}[\alpha])=\sigma(\tilde{H}_\omega)$.

 If we set $\tilde{H}_{\omega,\tilde{\omega}}[+\infty]\equiv\tilde{H}_\omega$, then the family $\{\tilde{H}_{\omega,\tilde{\omega}}[\alpha],\ \alpha\in(-\infty,+\infty]\}$ fulfills the assumptions of Theorem \ref{theo:abstract}. In particular, for all $\alpha>x_0$,
 $$\min\left\{1,\frac{\tilde{\omega}}{\omega}\right\}^2\tilde{H}_\omega\leq \tilde{H}_{\omega,\tilde{\omega}}[\alpha],$$
 and so the operator family obeys \eqref{prop:low_bound} of the theorem. (Remark that if $\omega \tilde{\omega}<0$, then \eqref{prop:low_bound} would not be fulfilled.) Therefore, for any $\mu\in\mathrm{Res}(\tilde{H}_\omega)$, we have 
 \begin{equation*}
  \lim_{\alpha\to +\infty}\|(\tilde{H}_{\omega,\tilde{\omega}}[\alpha]+\mu)^{-1}-(\tilde{H}_\omega+\mu)^{-1}\|=0.
 \end{equation*}
 Due to unitarity of $U_\alpha$ we also have
 \begin{equation*}
  \lim_{\alpha\to +\infty}\|(H_{\omega,\tilde{\omega}}[\alpha]+\mu)^{-1}-(H_\omega[\alpha]+\mu)^{-1}\|=0.
 \end{equation*}
 Since the norm-resolvent convergence implies the convergence of eigenvalues \cite{kato}, this yields
 \begin{proposition}\label{prop:spec_conv}
  Let $\sigma_n,\ n\in\N,$ be the $n$th eigenvalue of $H_{\omega,\tilde{\omega}}[\alpha]$, then 
  $$\lim_{\alpha\to +\infty}\sigma_n=(2n-1)\omega.$$
 \end{proposition}

\subsection{Proof of Theorem \ref{theo:main}}

 Let $\underline{H}_{\varepsilon,\pm}[\xi]$ and $\overline{H}_{\varepsilon,\pm}[\xi]$ be closures of $\dot{\underline{H}}_{\varepsilon,\pm}[\xi]$ and $\dot{\overline{H}}_{\varepsilon,\pm}[\xi]$, respectively, that are defined on $C_{0}^{\infty}(\R)$ by
 \begin{align*}
  &\dot{\underline{H}}_{\varepsilon,\pm}[\xi]=-d_{x}^{2}+\underline{V}_{\varepsilon,\pm}+\underline{W}_{\pm}-\varepsilon\\
  &\dot{\overline{H}}_{\varepsilon,\pm}[\xi]=-d_{x}^{2}+\overline{V}_{\varepsilon,\pm}+\overline{W}_{\pm}+\varepsilon.
 \end{align*}
 Then $\underline{H}_{\varepsilon,\pm}[\xi],\ \overline{H}_{\varepsilon,\pm}[\xi]$ are selfadjoint and have the structure of the comparison operator of the subsection \ref{sec:comparison} with $x_\xi$ being the free parameter instead of $\alpha$. (For the case $\xi\to+\infty$, we have $x_\xi\to -\infty$, and so the results of the subsection \ref{sec:comparison} must be modified in an obvious manner.)
 
 By Lemma \ref{lem:pot_est}, for all $\xi<\xi_\varepsilon$,
 $$\underline{H}_{\varepsilon,+}[\xi]\leq H[\xi]\leq \overline{H}_{\varepsilon,+}[\xi],$$
 and for all $\xi>\tilde{\xi}_\varepsilon$,
 $$\underline{H}_{\varepsilon,-}[\xi]\leq H[\xi]\leq \overline{H}_{\varepsilon,-}[\xi].$$
 If we now apply the minimax principle together with Proposition \ref{prop:spec_conv}, we obtain
 \begin{align*}
  &(\underline{B}_{\pm}-2\varepsilon)(2n-1)+\underline{W}_{\pm}-\varepsilon\leq\liminf_{\xi\to \mp\infty}\lambda_n[\xi]\\
  & \limsup_{\xi\to \mp\infty}\lambda_n[\xi]\leq (\overline{B}_{\pm}+2\varepsilon)(2n-1)+\overline{W}_{\pm}+\varepsilon.
 \end{align*}
 Since $\varepsilon$ may be arbitrarily small,
 \begin{align*}
  &\underline{B}_{\pm}(2n-1)+\underline{W}_{\pm}\leq\liminf_{\xi\to \mp\infty}\lambda_n[\xi]\\
  & \limsup_{\xi\to \mp\infty}\lambda_n[\xi]\leq \overline{B}_{\pm}(2n-1)+\overline{W}_{\pm}.
 \end{align*}
 
 Therefore if, for all $n\in\N$, either
 $$\overline{B}_{-}(2n-1)+\overline{W}_{-}<\underline{B}_{+}(2n-1)+\underline{W}_{+}$$
 or
 $$\overline{B}_{+}(2n-1)+\overline{W}_{+}<\underline{B}_{-}(2n-1)+\underline{W}_{-},$$
 then every $\lambda_n[\xi]$ is non-constant in $\xi$-variable. This  may rewritten as \eqref{eq:ac_cond}.

\section{example--effective Hamiltonian for a thin curved quantum layer in homogeneous magnetic field}\label{sec:example}
The quantum layers are important representatives of the so-called quantum wave\-guides that have been extensively studied over last several decades. See a recent monograph \cite{ExKo} for an immense list of references. The quantum waveguides in magnetic field, that will be of our particular interest, were examined, e.g., in \cite{KrRaTu_15, Gr_08,KrRa_14,BeOlVe_14,EkKo_05,Ol_14,BrRaSo_08}. In this section, we will derive a sufficient condition for the absolute continuity of the effective Hamiltonian for a very thin curved quantum layer in an ambient homogeneous magnetic field.

Let $\Sigma$ be a $y$-translation invariant surface in $\R^3$ given by the following parametrization:
$$\mathscr{L}_{0}(s,y)=(x(s),y,z(s))$$
with $s,y\in\R$. Here the functions $x$ and $z$ are assumed to be $C^{3}$-smooth and  such that $\dot x(s)^2+\dot z(s)^2=1$. The latter condition  means that the curve $\Gamma:s\mapsto(x(s),z(s))$ in $xz$ plane is parametrized by arc length measured from some reference point on the curve. Therefore the curvature $\kappa$ of $\Gamma$ is given by
$$\kappa(s)^2=\ddot x(s)^2+\ddot z(s)^2$$ 
and a unit normal vector to $\Sigma$ may be chosen as follows,
$$n(s,y)\equiv n(s)=(-\dot z(s),0,\dot x(s)).$$
If we view $\Sigma$ as a Riemannian manifold then the metric induced by the immersion $\mathscr{L}_{0}$ reads
\begin{equation*}
 (g_{\mu\nu})=\begin{pmatrix}
               1 & 0\\
               0 & 1
              \end{pmatrix}.
\end{equation*}

Let $a>0$ and $I:=(-1,1)$. Define a layer $\Omega$ of the width $2a$ constructed along $\Sigma$ as the image of
$$\mathscr{L}:\R^2\times I\to\R^3:\ \left\{ (s,y,u)\mapsto \mathscr{L}_{0}(s,y)+a u n(s)\right\}.$$
We always assume  $a<\|\kappa\|_{\infty}^{-1}$ and that $\Omega$ does not intersect itself.
Under these conditions, $\mathscr{L}$ is a diffeomorphism onto $\Omega$ as one can see, e.g., from the formula for the metric $G$ (induced by the immersion $\mathscr{L}$) on $\Omega$ that reads
\begin{equation*}
 (G_{ij})=\begin{pmatrix}
           (G_{\mu\nu})& 0\\
           0 & a^2
          \end{pmatrix},\quad
 (G_{\mu\nu})=\begin{pmatrix}
               f_a(s,u)^2 & 0\\
               0 & 1
              \end{pmatrix},
\end{equation*}
where $f_a(s,u):=(1-a u\kappa(s))$.

We start with the magnetic Laplacian on $\Omega$ subject to the Dirichlet boundary condition,
$$-\Delta_{D,A}^{\Omega}=(-i\nabla+A)^2\text{ (in the form sense)},\quad Q(-\Delta_{D,A}^{\Omega})=\mathcal{H}_{A,0}^{1}(\Omega,\dd x\dd y\dd z),$$
with a special choice of the vector potential, $A=B_{0}(0,x,0),\ B_{0}>0$, that corresponds to the magnetic field $B=(0,0,B_{0})$.  Employing the diffeomorphism $\mathscr{L}$, we may identify $-\Delta_{D,A}^{\Omega}$ with a selfadjoint operator $\hat{H}$ on $L^{2}(\R^2\times I,\dd\Omega)$ with the following action (understood in the form sense)
\begin{equation*}
\hat{H}_\Omega=-f_a(s,u)^{-1}\partial_{s}f_a(s,u)^{-1}\partial_{s}+(-i\partial_{y}+\tilde{A}_{2}(s,u))^2-a^{-2}f_a(s,u)^{-1}\partial_{u}f_a(s,u)\partial_{u},
\end{equation*}
where $\tilde{A}=(D\mathscr{L})^T A\circ\mathscr{L}=(0,\tilde{A}_{2},0)$ with
$\tilde{A}_{2}(s,u)=B_{0}\big(x(s)-a u\dot z(s)\big).$

Involving a unitary transform $U:L^{2}(\R^2\times I,\dd\Omega)\to L^{2}(\R^2\times I,\dd \Sigma\dd u)$, $\psi\mapsto a^{1/2}f_{a}^{1/2}\psi$ we arrive at a unitarily equivalent operator defined again in the form sense as
\begin{equation*}
 \tilde{H}_\Omega=U\hat{H}_\Omega U^{-1}=-\partial_{s}f_a(s,u)^{-2}\partial_{s}+(-i\partial_{y}+\tilde{A}_{2}(s,u))^2-a^{-2}\partial_{u}^{2}+V(s,u)
\end{equation*}
where 
$$V(s,u)=-\frac{1}{4}\frac{\kappa(s)^2}{f_a(s,u)^2}-\frac{1}{2}\frac{a u\ddot\kappa(s)}{f_a(s,u)^3}-\frac{5}{4}\frac{a^2 u^2\dot\kappa(s)^2}{f_a(s,u)^4}.$$
(We have to strengthen our regularity assumptions on $\Sigma$ to give a meaning to the second derivative of the curvature.)

It was proved in \cite{KrRaTu_15} that, for all $k$ large enough,
$$\|\big(\tilde{H}_\Omega-(\pi/2a)^2+k\big)^{-1}-(h_{\mathrm{eff}}+k)^{-1}\oplus 0\|=\mathcal{O}(a)$$
as $a\to 0_{+}$, with 
\begin{equation*}
 h_{\mathrm{eff}}=-\partial_{s}^{2}+(-i\partial_{y}+B_{0}x(s))^2-\frac{1}{4}\kappa^{2}(s)
\end{equation*}
acting on $L^2(\R^2,\dd s\dd y)$.
If we assume that $\kappa$ is bounded, then $h_{\mathrm{eff}}$ is essentially selfadjoint on $C_{0}^{\infty}(\R^2)$.

Clearly, $h_{\mathrm{eff}}$ is of the form \eqref{eq:ham_def}. Theorem \ref{theo:main}  yields immediately 
\begin{proposition}
$h_{\mathrm{eff}}$ is purely absolutely continuous if $\kappa\in L^\infty$ and either 
$$\underline{\dot{x}}_{\pm}>0\,\wedge\,\underline{\dot{x}}_{+}\geq\overline{\dot{x}}_{-}\,\wedge\,\left(\underline{\kappa^2}_{+}-\overline{\kappa^2}_{-}<4 B_0(\underline{\dot{x}}_{+}-\overline{\dot{x}}_{-})\right)$$
or
$$\underline{\dot{x}}_+>0\,\wedge\, \overline{\dot{x}}_-<0.$$
(The claim remains valid if we interchange the $\pm$ indices everywhere.)
In particular, it is purely absolutely continuous if
$$\lim_{s\to\pm\infty}\kappa(s)=0\ \text{ and } \lim_{s\to\infty}\dot{x}(s)\neq \lim_{s\to-\infty}\dot{x}(s).$$
\end{proposition}

\section{Appendix: Proof of Theorem \ref{theo:abstract}}\label{sec:abstract_result}

For our proof we will need two auxiliary results. The first one is due to C.~Cazacu and D.~Krej\v{c}i\v{r}\'{i}k \cite{CaKr_16}. We present it here in a refined form.

\begin{lemma}[Cazacu, Krej\v{c}i\v{r}\'{i}k \cite{CaKr_16}]\label{lem:CaKr}
 Let $\{R_d\}_{d\in\R}$ be a family of bounded operators and $R$ be a compact operator on some Hilbert space.
 If for all sequences $(f_n)$ with properties $\|f_n\|=1$ and $f_n\xrightarrow[n\to+\infty]{w}f$, any real sequence $(d_n)$ such that $\lim_{n\to+\infty} d_n=+\infty$, and any $\varepsilon>0$, there exists a subsequence $(n_k)$  such that $\lim_{k\to+\infty}\|R_{d_{n_k}}f_{n_k}-Rf\|<\varepsilon$, then $\lim_{d\to+\infty}\|R_d-R\|=0$.
\end{lemma}

\begin{lemma}\label{lem:abstract}
 Let $\{A[\alpha]\}$ be as in Theorem \ref{theo:abstract} and $\mu$ be such that $(A[\alpha]+\mu)\geq 1$ for all $\alpha\in[\alpha_0,+\infty]$. ($\mu$ with this property exists, due to \eqref{prop:low_bound} and semiboundness of $A[+\infty]$.)
 Then for any sequence of functions $(f_n)$ such that $\|f_n\|=1$ and $f_n\xrightarrow[n\to+\infty]{w}f$, any real sequence $(d_n)|\, \lim_{n\to+\infty} d_n=+\infty$, and any $\varepsilon>0$, there exists a subsequence $(n_k)$  such that
 $$\lim_{k\to+\infty}\|(A[d_{n_k}]+\mu)^{-1}f_{n_k}-(A[+\infty]+\mu)^{-1}f\|<\varepsilon.$$
\end{lemma}

\begin{proof}
  Let $u_{n}$ be uniquely defined by
 $$(A[d_n]+\mu)u_{n}=f_{n}.$$
 Due to the hypothesis \eqref{prop:core} of Theorem \ref{theo:abstract}, we may construct sequences $(\tilde{u}_n)$ and $(\tilde{f}_n)$ with properties
 \begin{equation*}
  \tilde{u}_n\in C_{0}^{\infty}(\Omega),\quad (A[d_n]+\mu)\tilde{u}_{n}=\tilde{f}_{n},\quad \|\tilde{u}_n-u_n\|<\varepsilon,\quad  \|\tilde{f}_n-f_n\|<\varepsilon,
 \end{equation*}
 for all $n\in\N$.
 
 Using the hypothesis \eqref{prop:low_bound} we obtain
 \begin{multline*}
  C \langle (A[+\infty]+\mu)^{1/2}\tilde{u}_{n},(A[+\infty]+\mu)^{1/2}\tilde{u}_{n}\rangle+(K+\mu(1-C)) \|\tilde{u}_n\|\leq\\
  \langle (A[d_n]+\mu)^{1/2}\tilde{u}_{n},(A[d_n]+\mu)^{1/2}\tilde{u}_{n}\rangle
  =\langle \tilde{u}_{n},\tilde{f}_{n}\rangle\leq \|\tilde{u}_{n}\| \|\tilde{f}_{n}\|=\|\tilde{u}_{n}\|(1+\varepsilon).
 \end{multline*}
 Since $\|\tilde{u}_{n}\|\leq \|(A[d_n]+\mu)^{-1}\|\,\|\tilde{f}_n\|\leq(1+\varepsilon)$,
 we deduce from here that $(\tilde{u}_{n})$ is bounded in the topology of $(A[+\infty]+\mu)^{1/2}$. Consequently, there is a weakly convergent subsequence $(\tilde{u}_{n_k})$ with respect to this topology, whose limit will be denoted by $\tilde{u}$.

 Consider $v\in C_{0}^{\infty}(\Omega)$. Then we have
 $$\langle (A[d_{n_k}]+\mu)^{1/2} v, (A[d_{n_k}]+\mu)^{1/2}\tilde{u}_{n_k}\rangle=\langle v, \tilde{f}_{n_k}\rangle,$$
 which, due to the hypothesis \eqref{prop:conv}  and the Urysohn lemma implies that, for all $k$ large enough,
 \begin{equation}\label{eq:form_eq}
  \langle (A[+\infty]+\mu)^{1/2} v, (A[+\infty]+\mu)^{1/2}\tilde{u}_{n_k}\rangle=\langle v, \tilde{f}_{n_k}\rangle.
 \end{equation}
 Since $(\tilde{f}_{n_k})$ is bounded, it has a weakly convergent subsequence, say $(\tilde{f}_{\tilde{n}_k})$.  We will abuse the notation a little and write just $n_k$ instead of $\tilde{n}_k$.
  Now, in the limit $k\to +\infty$, \eqref{eq:form_eq} yields
 \begin{equation}\label{eq:form_eq_2}
 \langle (A[+\infty]+\mu)^{1/2} v, (A[+\infty]+\mu)^{1/2}\tilde{u}\rangle=\langle v, \tilde{f}\rangle.
 \end{equation}
 
 Since $C_{0}^{\infty}(\Omega)$ is a core of $(A[+\infty]+\mu)$ and any core of a selfadjoint operator is a form core too, \eqref{eq:form_eq_2} implies that $\tilde{u}\in\dom{A[+\infty]}$ and $(A[+\infty]+\mu)\tilde{u}=\tilde{f}$.  Moreover, $(A[+\infty]+\mu)\tilde{u}_{n_k}-\tilde{f}_{n_k}\xrightarrow[k\to\infty]{w}0$ by \eqref{eq:form_eq} and $(A[d_{n_k}]+\mu)\tilde{u}_{n_k}=\tilde{f}_{n_k}\xrightarrow[k\to\infty]{w}\tilde{f}$, which yields $(A[+\infty]+\mu)\tilde{u}_{n_k}\xrightarrow[k\to\infty]{w}(A[+\infty]+\mu)\tilde{u}$. 
 Finally, by compactness of $(A[+\infty]+\mu)^{-1}$, $\tilde{u}_{n_k}\xrightarrow[k\to\infty]{s}\tilde{u}$, i.e.,
 $$\lim_{k\to+\infty}\|\tilde{u}_{n_k}-\tilde{u}\|=\lim_{k\to+\infty}\|(A[d_{n_k}]+\mu)^{-1}\tilde{f}_{n_k}-(A[+\infty]+\mu)^{-1}\tilde{f}\|=0.$$
 
 Coming back to the untilded sequences, for all $k$ large enough, we obtain
 \begin{multline*}
  \|(A[d_{n_k}]+\mu)^{-1}f_{n_k}-(A[+\infty]+\mu)^{-1}f\|\leq \|(A[d_{n_k}]+\mu)^{-1}\tilde{f}_{n_k}-(A[+\infty]+\mu)^{-1}\tilde{f}\|\\
  +\|(A[d_{n_k}]+\mu)^{-1}(f_{n_k}-\tilde{f}_{n_k})\|+\|(A[+\infty]+\mu)^{-1}(\tilde{f}-f)\|\\
  \leq \|(A[d_{n_k}]+\mu)^{-1}\tilde{f}_{n_k}-(A[+\infty]+\mu)^{-1}\tilde{f}\|+\|f_{n_k}-\tilde{f}_{n_k}\|+\|\tilde{f}-f\|.
 \end{multline*}
 The limit of the first term is zero, the second term is bounded by $\varepsilon$, and the third term is bounded by $\varepsilon$, too, since $\|\tilde{f}-f\|^2=\lim_{k\to +\infty}\langle\tilde{f}-f, \tilde{f}_{n_k}-f_{n_k}\rangle\leq \|\tilde{f}-f\| \lim_{k\to +\infty}\|\tilde{f}_{n_k}-f_{n_k}\|\leq \varepsilon\|\tilde{f}-f\|.$
\end{proof}

\begin{proof}[Proof of Theorem \ref{theo:abstract}]
 For $z=-\mu$ with $\mu$ specified in Lemma \ref{lem:abstract}, the theorem follows immediately. For $z\in\mathrm{Res}(A[+\infty])$, we use formula \cite[(3.10) of Chap. 4]{kato} to extend the resolvent estimate.
\end{proof}

\begin{remark}[Alternative proof of Theorem \ref{theo:abstract}] During the peer review process one of the reviewers proposed an alternative proof that actually does not require $L^2$-setting of Theorem \ref{theo:abstract}. It is based on some properties of collectively compact operator sequences. A set of operators is called collectively compact if and only  the union of the images of the unit ball is precompact \cite{AnPa_68}. The hypotheses \eqref{prop:low_bound} and \eqref{prop:comp} of Theorem \ref{theo:abstract} imply that $\{(A[\alpha]-i)^{-1},\, \alpha\geq\alpha_0\}$ is collectively compact. Indeed, let $\mathscr{B}$ be the closed unit ball (in $L^2(\Omega)$) and $M:=\cup_{\alpha\geq\alpha_0}\{(A[\alpha]-i)^{-1}\mathscr{B}\}$. If we take $\psi\in M$ then there exists $\alpha\geq\alpha_0$ such that $\psi=(A[\alpha]-i)^{-1}\varphi$ for some $\varphi\in L^2(\Omega)$ with $\|\varphi\|\leq 1$. Using the functional calculus, we infer that $\|(A[\alpha]-i)^{-1}\|\leq 1$ and $\|A(A-i)^{-1}\|\leq 1$. 
Consequently, $\|\psi\|\leq 1$ and $\langle\psi, A[\alpha]\psi\rangle\leq \|(A[\alpha]-i)^{-1}\| \|A(A-i)^{-1}\|\|\varphi\|^2\leq 1.$ If we introduce sets $\mathscr{A}_{\alpha,b}:=\{\psi\in Q(A[\alpha]):\, \|\psi\|\leq 1,\, \langle \psi, A[\alpha]\psi\rangle\leq b\}$, we may write $\psi\in\mathscr{A}_{\alpha,1}$. By \eqref{prop:low_bound} of Theorem \ref{theo:abstract}, 
$\psi\in\mathscr{A}_{+\infty,(1-K)/C}$. We conclude that $M\subset\mathscr{A}_{+\infty,(1-K)/C}$. The latter set is compact due to \eqref{prop:comp} of Theorem \ref{theo:abstract} and \cite[Theo. XIII.64]{RS4}. Therefore, $M$ is precompact.

Using the hypotheses \eqref{prop:core} and \eqref{prop:conv} of Theorem \ref{theo:abstract} together with \cite[Corollary VIII.1.6]{kato}, we deduce that $(A[\alpha]-i)^{-1}\xrightarrow[\alpha\to +\infty]{s}(A[+\infty]-i)^{-1}$. By \cite[Theo. 3.4 and Prop. 2.1]{AnPa_68}, the strong-resolvent convergence together with the collective compactness imply the norm-resolvent convergence, i.e., we have $\lim_{\alpha\to +\infty}\|(A[\alpha]-i)^{-1}-(A[+\infty]-i)^{-1}\|=0$.
\end{remark}

\section{Acknowledgement}
The author wishes to express his thanks to P. Exner for drawing his attention to the Iwatsuka model and to D. Krej\v{c}i\v{r}\'{i}k for pointing out some of his results that were useful in Section \ref{sec:abstract_result}. Last but not least, the author gratefully acknowledges many useful suggestions of one of the reviewers that helped improve the manuscript significantly. The work has been supported by the grant No. 13--11058S of the Czech Science Foundation (GA\v{C}R).

\end{document}